%% file: sensors.tex
\documentclass{llncs}

\usepackage{amsmath,amssymb}
\usepackage{ifthen}
\usepackage{xspace}
\usepackage{paralist,mdwlist}
\usepackage{comment}

\usepackage{algorithm}
\usepackage{algorithmicx}
\usepackage{algpseudocode}

\usepackage{color,soul}

\usepackage{tikz}

\usepackage{wrapfig}

\spnewtheorem{observation}[lemma]{Observation}{\bfseries}{\itshape}

\newcommand{\eqdf}{\stackrel{\rm def}{=}}
\newcommand{\set}[1]{\left\{ #1 \right\}}
\newcommand{\paren}[1]{\left( #1 \right)}

\newcommand{\inv}[1]{\frac{1}{#1}}

\newcommand{\ceil}[1]{\left\lceil {#1} \right\rceil}

\newcommand{\half}{\frac{1}{2}}

\newcommand{\argmax}{\operatornamewithlimits{argmax}}

\newcommand{\bcfr}{\textsc{BCFR}\xspace}
\newcommand{\bcvr}{\textsc{BCVR}\xspace}

\newcommand{\partition}{\textsc{Partition}\xspace}

\newcommand{\threepart}{\textsc{3-Partition}\xspace}

\newcommand{\eps}{\varepsilon}

\newcommand{\fixed}{\textbf{Fixed}\xspace}
\newcommand{\variable}{\textbf{Variable}\xspace}

\begin{document}

\title{
Maximizing Barrier Coverage Lifetime \\ with Mobile Sensors}

\author{%
Amotz Bar-Noy\inst{1}
\and 
Dror Rawitz\inst{2}
\and 
Peter Terlecky\inst{1}
}

\institute{%
The Graduate Center of the City University of New York, NY 10016, USA \\
\email{amotz@sci.brooklyn.cuny.edu}, \email{pterlecky@gc.cuny.edu}
\and
School of Electrical Engineering, Tel Aviv University, Tel-Aviv 69978, 
Israel \\
\email{rawitz@eng.tau.ac.il}
}

\maketitle

\begin{abstract}
Sensor networks are ubiquitously used for detection and tracking and
as a result covering is one of the main tasks of such networks.  We
study the problem of maximizing the \emph{coverage lifetime} of a
barrier by mobile sensors with limited battery powers, where the
coverage lifetime is the time until there is a breakdown in coverage
due to the death of a sensor.
Sensors are first deployed and then coverage commences.  Energy is
consumed in proportion to the distance traveled for mobility, while
for coverage, energy is consumed in direct proportion to the radius of
the sensor raised to a constant exponent.
We study two variants which are distinguished by whether the sensing
radii are given as part of the input or can be optimized, the fixed
radii problem and the variable radii problem.
We design parametric search algorithms for both problems for the case
where the final order of the sensors is predetermined 
and for the case where sensors are initially located at barrier
endpoints.
In contrast, we show that the variable radii problem is strongly
NP-hard and provide hardness of approximation results for fixed radii
for the case where all the sensors are initially co-located at an
internal point of the barrier.
\end{abstract}



\input{introduction}

\input{model}

\input{search}

\input{zero}

\input{conclusion}


\bibliographystyle{abbrv}
\bibliography{sensors}


\newpage
\appendix

\input{friction}

\input{proofs}

\input{hardness}

\end{document}

%% file: introduction.tex
\section{Introduction}

One important application of Wireless Sensor Networks is monitoring a
barrier for some phenomenon.  By covering the barrier, the sensors
protect the interior of the region from exogenous elements more
efficiently than if they were to cover the interior area.  
In this paper we focus on a model in which sensors are battery-powered
and both moving and sensing drain energy.  A sensor can maintain
coverage until its battery is completely depleted.  The network of
sensors cover the barrier until the death of the first sensor, whereby
a gap in coverage is created and the life of the network expires.

More formally, there are $n$ sensors denoted by $\{1,\ldots,n\}$.  Each
sensor $i$ 
has a battery of size $b_i$ and initial position $x_i$.  The coverage
task is accomplished in two phases.  In the \emph{deployment phase},
sensors move from their initial positions to new positions, and in the
\emph{covering phase} the sensors set their sensing radii to fully
cover the barrier.
A sensor which moves a distance $d$ drains $a \cdot d$ amount of
battery on movement for some constant $a \geq 0$.  In the coverage
phase, sensing with a radius of $r$ drains energy per time unit in
direct proportion to $r^\alpha$, for some constant $\alpha \geq 1$
(see e.g.,~\cite{AGMP09,Li}).  The lifetime of a sensor $i$ traveling
a distance $d_i$ and sensing with a radius $r_i$ is given by $L_i=
\frac{b_i-a d_i}{r_i^\alpha}$.  The coverage lifetime of the barrier
is the minimum lifetime of any sensor, $\min_i L_i$.  We seek to
determine a destination $y_i$ and a radius $r_i$, for each sensor $i$,
that maximizes the \emph{barrier coverage lifetime} of the network.

Many parameters govern the length of coverage lifetime, and optimizing
them is hard even for simple variants.  Therefore, most of the past
research adopted natural strategies that try to optimize the lifetime
indirectly.  For example, the duty cycle strategy partitions the
sensors into disjoint groups that take turns in covering the barrier.
The idea is that a good partition would result in a longer lifetime.
Another example is the objective of minimizing the maximum distance
traveled by any of the sensors.  This strategy would maximize the coverage
lifetime for sensors with homogeneous batteries and radii, but would
fail to do so if sensors have non-uniform batteries or radii.  See a
discussion in the related work section.

In this paper we address the lifetime maximization problem directly.
We focus on the \emph{set-up and sense} model in which the sensors are
given one chance to set their positions and sensing radii before the
coverage starts.  
We leave the more general model in which sensors may
adjust their positions and sensing radii during the coverage to future
research.


\paragraph*{\bf Related work.}
There has been previous research on barrier coverage focused on
minimizing a parameter which is proportional to the energy sensors
expend on movement, but not directly modeling sensor lifetimes with
batteries.  Czyzowicz \textit{et al.}~\cite{czyzowicz2009minmax}
assume that sensors are located at initial positions on a line barrier
and that the sensors have fixed and identical sensing radii.  The goal
is to find a deployment that covers the barrier and that minimizes the
maximum distance traveled by any sensor.  Czyzowicz \textit{et al.} 
provide a polynomial time algorithm for this problem.  Chen \textit{et
al.}~\cite{barrier} extended the result to the more general case in
which the sensing radii are non-uniform (but still fixed).

Czyzowicz \textit{et al.}\cite{czyzowicz2010minsum} considered
covering a line barrier with sensors with the goal of minimizing the
sum of the distances traveled by all sensors.  Mehrandish \textit{et
al.}~\cite{Mehrandish2011} considered the same model with the
objective of minimizing the number of sensors which must move
to cover the barrier.  Tan and Wu~\cite{Tan} presented improved
algorithms for minimizing the max distance traveled and minimizing the
sum of distances traveled when sensors must be positioned on a circle
in regular $n$-gon position.  The problems were initially considered
by Bhattacharya \textit{et al.}~\cite{bhat}.
Several works have considered the problem of covering a straight-line
boundary by stationary sensors. Li \textit{et al.}~\cite{Li} look to
choose radii for sensors for coverage which minimize the sum of the
power spent.  Agnetis \textit{et al.}~\cite{AGMP09} seek to choose
radii for coverage to minimize the sum of a quadratic cost function.
Maximizing the network lifetime of battery-powered sensors that cover
a barrier was previously considered for static sensors from a
scheduling point of view.  Buchsbaum \textit{et al.}~\cite{BEJVY07}
and Gibson and Varadarajan~\cite{GibVar09} considered the
\textsc{Restricted Strip Covering} in which sensors are static and
radii are fixed, but sensors may start covering at any time.  Bar Noy
\textit{et al.}~\cite{BarBau11,BBR12,BBR12corr} considered the variant
of this problem in which the radii are adjustable.

The only previous result we are aware of that considered a battery model with
movement and transmission on a line is by Phelan \textit{et
al.}~\cite{Phelan} who considered the problem of maximizing the
transmission lifetime of a sender to a receiver on a line using mobile
relays.   


\paragraph{\bf Our contribution.}
We introduce two problems in the model in which sensors are
battery-powered and both moving and sensing drain energy.  In the
\textsc{Barrier Coverage with Variable Radii} problem (abbreviated
\bcvr) we are given initial locations and battery powers, and the goal
is to find a deployment and radii that maximizes the lifetime.  In the
\textsc{Barrier Coverage with Fixed Radii} problem (\bcfr) we are also
given a radii vector $\rho$, and the goal is to find a deployment and
a radii assignment $r$, such that $r_i \in
\set{0,\rho_i}$, for every $i$, that maximizes the lifetime.

We show in Appendix~\ref{sec:friction} that the static ($a=\infty$)
and fully dynamic ($a=0$) cases are solvable in polynomial time for
both \bcfr and \bcvr.

In Section~\ref{sec:search} we consider constrained versions of \bcfr
and \bcvr in which the input contains a total order on the sensors
that the solution is required to satisfy.  We design polynomial-time
algorithms for the decision problems in which the goal is to determine
whether a given lifetime $t$ is achievable and to compute a solution
with lifetime $t$, if $t$ is achievable.  Using these decision
algorithms we present parametric search algorithms for constrained
\bcfr and \bcvr.

We consider the case where the sensors are initially located on the
edges of the barrier (i.e., $x \in \set{0,1}^n$) in
Section~\ref{sec:zero}.  For both \bcfr and \bcvr, we show that, for
every candidate lifetime $t$, we may assume a final ordering of the
sensors.  (The ordering depends only on the battery powers in the
\bcvr case, and it can be computed in polynomial time in the \bcfr
case.)  Using our decision algorithms, we obtain parametric search
algorithms for this special case.


On the negative side, we show that there is no polynomial time
multiplicative approximation algorithm for \bcfr and that there is no
polynomial time algorithm that computes solutions that are within an
additive factor $\eps$, for some constant $\eps>0$, unless P$\neq$NP.
Both results hold even if $x = p^n$, for some $p \in (0,1)^n$.  We
also show that \bcvr is strongly NP-hard.  The hardness results apply
to any $0 < a < \infty$ and $\alpha \geq 1$ and they are given in
Section~\ref{sec:hardness}.

Finally, we note that several proofs 
were relegated to the appendix due to space considerations.

%% file: model.tex
\section{Preliminaries}

In this section we formally define the problems and introduce the
notation that will be used throughout the paper.

\paragraph{\bf Model.}
We consider a setting in which $n$ mobile sensors with finite
batteries are located on a barrier represented by the interval
$[0,1]$.  The initial position and battery power of sensor $i$ is
denoted by $x_i$ and $b_i$, respectively.  We denote
$x=(x_1,\ldots,x_n)$ and $b=(b_1,\ldots,b_n)$.
%
The sensors are used to cover the barrier, and they can achieve this
goal by moving and sensing.  In our model the sensors first move, and
afterwards each sensor covers an interval that is determined by its
sensing radius.  In motion, energy is consumed in proportion to the
distance traveled, namely a sensor consumes $a \cdot d$ units of
energy by traveling a distance $d$, where $a$ is a constant.  A sensor
$i$ consumes $r_i^\alpha$ energy per time unit for sensing, where
$r_i$ is the sensor's radius and $\alpha \geq 1$ is a constant.

More formally, the system works in two phases.  In the
\emph{deployment phase} sensors move from the initial positions $x$ to
new positions $y$.  This phase is said to occur at time $0$.  In this
phase, sensor $i$ consumes $a|y_i-x_i|$ energy.  Notice that sensor
$i$ may be moved to $y_i$ only if $a |y_i-x_i| \leq b_i$.  
In the \emph{covering phase} sensor $i$ is assigned a sensing radius
$r_i$ and covers the interval $[y_i-r_i,y_i+r_i]$.  (An example is
given in Figure~\ref{fig:model}.)
A pair $(y,r)$, where $y$ is a deployment vector and $r$ is a
sensing radii vector, is called \emph{feasible} if
\begin{inparaenum}[(i)]
\item $a |y_i-x_i| \leq b_i$, for every sensor $i$, and 
\item $[0,1] \subseteq \sum_i [y_i-r_i,y_i+r_i]$.
\end{inparaenum}
Namely, $(y,r)$ is feasible, if the sensors have enough power to reach
$y$ and each point in $[0,1]$ is covered by some sensor.

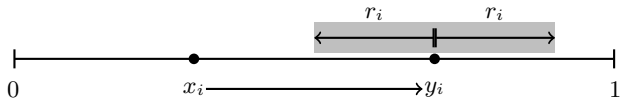
\begin{figure}[t]
\centering
\begin{tikzpicture}[scale=0.8,auto,thick,fill=lightgray,
dot/.style={circle,fill=black,inner sep=0pt,minimum size=4pt}]
  \draw (0,0)[arrows=|-|] -- (10,0);
  \draw (0,-0.5) node {$0$};
  \draw (10,-0.5) node {$1$};
  \node at (3,0) [dot] {};
  \draw (3,-0.5) node {$x_i$};
  \fill (5,0.1) rectangle (9,0.6);
  \draw[arrows=|->] (7,0.35) -- (9,0.35);
  \draw (8,0.75) node {$r_i$};
  \draw[arrows=|->] (7,0.35) -- (5,0.35);
  \draw (6,0.75) node {$r_i$};
  \node at (7,0) [dot] {};
  \draw (7,-0.5) node {$y_i$};
  \draw[arrows=->] (3.2,-0.5) -- (6.8,-0.5);
\end{tikzpicture}
\vspace{-5pt}
\label{fig:model}
\caption{Sensor $i$ moves from $x_i$ to $y_i$ and covers the 
interval $[y_i-r_i,y_i+r_i]$.}
\end{figure}

Given a feasible pair $(y,r)$, the \emph{lifetime} of a sensor $i$,
denoted $L_i(y,r)$, is the time that transpires until its battery is
depleted.  If $r_i>0$,
\(
L_i(y,r) = \frac{b_i-a|y_i-x_i|}{r_i^\alpha}
\), 
and if $r_i = 0$, we define $L_i(y,r) = \infty$.
Given initial locations $x$ and battery powers $b$, the
\textit{barrier coverage lifetime} of a feasible pair $(y,r)$, where
$y$ is a deployment vector and $r$ is a sensing radii vector is
defined as
\(
L(y,r) = \min_i L_i(y,r)
\).
We say that a $t$ is \emph{achievable} if there exists a feasible pair
such that $L_i(y,r) = t$.



\paragraph*{\bf Problems.}
We consider two problems which are distinguished by whether the radii
are given as part of the input.
In the \textsc{Barrier Coverage with Variable Radii} problem (\bcvr)
we are given initial locations $x$ and battery powers $b$, and the
goal is to find a feasible pair $(y,r)$ of locations and radii that
maximizes $L(y,r)$.
In the \textsc{Barrier Coverage with Fixed Radii} problem (\bcfr) 
we are also given a radii vector $\rho$, and the goal is to find a
feasible pair $(y,r)$, such that $r_i \in \set{0,\rho_i}$ for every
$i$, that maximizes $L(y,r)$.
Notice that a necessary condition for achieving non-zero lifetime is
$\sum_i 2\rho_i \geq 1$.


Given a total order $\prec$ on the sensors, we consider the
\emph{constrained} variants of \bcvr and \bcfr, in which the
deployment $y$ must satisfy the following requirement: $i \prec j$ if
and only if $y_i \leq y_j$.  That is, we are asked to maximize barrier
coverage lifetime subject to the condition that the sensors are
ordered by $\prec$.
Without loss of generality, we assume that the sensors are numbered
according to the total order.


%

%% file: search.tex
\section{Constrained Problems and Parametric Search}
\label{sec:search}

We present polynomial time algorithms that, given $t > 0$, decide
whether $t$ is achievable for constrained \bcfr and constrained \bcvr.
If $t$ is achievable, a solution with lifetime at least $t$ is
computed.  We use these algorithms to design parametric search
algorithms for both problems.

We use the following definitions for both \bcfr and \bcvr.  Given an
order requirement $\prec$, we define:
\begin{align*}
l(i) & \textstyle
       \eqdf \max\set{\max_{j \le i} \set{x_j - b_j/a},0}
&  
u(i) & \textstyle
       \eqdf \min\set{\min_{j \ge i} \set{x_j + b_j/a},1}
\end{align*}
$l(i)$ and $u(i)$ are the leftmost and rightmost points reachable by
$i$.

\begin{observation}
\label{obs:lu}
Let $(y,r)$ be a feasible solution that satisfies an order requirement
$\prec$.  Then $l(i) \leq u(i)$ and $y_i \in [l(i),u(i)]$, for every
$i$.
\end{observation}
\begin{proof}
If there exists $i$ such that $u(i)<l(i)$, then there are two sensors
$j$ and $k$, such that where $k < j$ and $x_j + b_j/a < x_k - b_k/a$.
Hence, no deployment that satisfies the total order exists.
\qed
\end{proof}


\subsection{Fixed Radii}

We start with an algorithm that solves the constrained \bcfr decision
problem.


Given a \bcfr instance and a lifetime $t$, we define
\begin{align*}
s(i) & \eqdf \max \set{x_i - (b_i - t \rho_i^\alpha)/a, l(i)}
& 
e(i) & \eqdf \min \set{x_i + (b_i - t \rho_i^\alpha)/a, u(i)}
\end{align*}
If $t \rho_i^\alpha \leq b_i$, then $s(i) \leq e(i)$.  Moreover $s(i)$
and $e(i)$ are the leftmost and rightmost points that are reachable by
$i$, if $i$ participates in the cover for $t$ time.  
%
($l(i)$ and $u(i)$ can be replaced by $l(i-1)$ and $u(i-1)$
in the above definitions.)

\begin{observation}
Let $(y,r)$ be a feasible pair with lifetime $t$ that satisfies 
an order~$\prec$.  For every 
$i$, if $r_i = \rho_i$,
it must be that $t \rho_i^\alpha \leq b_i$ and $y_i \in
[s(i),e(i)]$.
\end{observation}

Algorithm~\fixed is our decision algorithm for constrained \bcfr.  It
first computes $l$, $u$, $s$, and $e$.  If there is a sensor $i$ such
that $l(i) > u(i)$, it outputs NO.  Otherwise it deploys the sensors
one by one according to $\prec$.  
Iteration $i$ starts with checking whether $i$ can extend the current
covered interval $[0,z]$.  If it cannot, $i$ is moved to the left as
much as possible (power is used only for moving), and it is powered
down ($r_i$ is set to 0).  If $i$ can extend the current covered
interval, it is assigned radius $\rho_i$, and it is moved to the
rightmost possible position, while maximizing the right endpoint of
the currently covered interval (i.e., $[0,z]$).  If $i$ is located to
the left of a sensor $j$, where $j<i$, then $j$ is moved to $y_i$.

\begin{algorithm}[t]
\begin{small}
\begin{algorithmic}[1]
\caption{: \fixed$(x,b,\rho,t)$}
\label{Alg:decide-fixed}
\State Compute $l$, $u$, $s$, and $e$
\State \textbf{if} there exists $i$ such that $u(i)<l(i)$
       \textbf{then} \Return{NO}
\State $z \gets 0$
\For{$i=1 \to n$}
   \If{$t \rho_i^\alpha > b_i$ or $z \not\in [s(i)-\rho_i,e(i)+\rho_i)$}
      \State $y_i \gets \max\set{l(i),y_{i-1}}$
             and $r_i \gets 0$
\Comment{$y_0 = 0$}
   \Else
      \State $y_i \gets \min\set{z+\rho_i,e(i)}$
             and $r_i \gets \rho_i$
%
%
      \State $S \gets \set{k : k < i, y_i < y_k}$
      \State $y_k \gets y_i$ and $r_k \gets 0$, for every $k \in S$
      \State $z \gets y_i + r_i$
   \EndIf
\EndFor
\State \textbf{if} $z<1$ \textbf{then} \Return{NO}
\State \textbf{else} \Return{YES}
\end{algorithmic}
\end{small}
\end{algorithm}

As for the running time, $l$, $u$, $s$ and $e$ can be computed in
$O(n)$ time.  There are $n$ iterations,
each takes $O(n)$ time. Hence, the running time of Algorithm~\fixed is
$O(n^2)$.
It remains to prove the correctness of the algorithm.

\begin{theorem}
\label{thm:search-bcfr}
\fixed solves the constrained \bcfr decision problem.
\end{theorem}
\begin{proof}
If $u(i)<l(i)$ for some $i$, then no deployment that satisfies the
order $\prec$ exists by Observation~\ref{obs:lu}.  Hence, the
algorithm responds correctly.

We show that if the algorithm outputs YES, then the computed solution
is feasible.  First, notice that $y_{i-1} \leq y_i$, for every $i$, by
construction.  We prove by induction on $i$, that $y_j
\in [l(j),u(j)]$ and that $y_j \in [s(j),e(j)]$, if $r_j = \rho_j$,
for every $j \leq i$.  
%
Consider the $i$th iteration.  If $t \rho_i^\alpha > b_i$ or $z
\not\in [s(i)-\rho_i,e(i)+\rho_i)$, then $y_i \in [l(i),u(i)]$, since
$\max\set{l(i),y_{i-1}} \leq \max\set{u(i),u(i-1)} \leq u(i)$.
Otherwise, $y_i = \min\set{z+\rho_i,e(i)} \geq s(i)$, since $z \geq
s(i) - \rho_i$.  Hence, if $r_i = \rho_i$, we have that $y_i \in
[s(i),e(i)]$.  Furthermore, if $j < i$ is moved to the left to $i$,
then $y_j = y_i \geq s(i) \geq l(i) \geq l(j)$.
Finally, let $z_i$ denote the value of $z$ after the $i$th iteration.
(Initially, $z_0 = 0$.)  We proof by induction on $i$ that $[0,z_i]$
is covered.  Consider iteration $i$.  If $r_i=0$, then we are done.
Otherwise, $z_{i-1} \in [y_i-\rho_i,y_i+\rho_i]$ and $z_i =
y_i+\rho_i$.  Furthermore, the sensors in $S$ can be powered down and
moved, since $[y_j-r_j,y_j+r_j] \subseteq [y_i-\rho_i,y_i+\rho_i]$, for
every $j \in S$.

Finally, we show that if the algorithm outputs NO, there is no
feasible solution.
We prove by induction that $[0,z_i$] is the longest interval than can
be covered by sensors $1,\ldots,i$.
In the base case, observe that $z_0=0$ is optimal.  
For the induction step, let $y'$ be a deployment of $1,\ldots,i$ that
covers the interval $[0,z'_i]$.  Let $[0,z'_{i-1}]$ be the interval
that $y'$ covers by $1,\ldots,i-1$.  By the inductive hypothesis,
$z'_{i-1} \leq z_{i-1}$.
If $t \rho_i^\alpha > b_i$ or $z_{i-1} < s(i) - \rho_i$, it follows
that $z'_i = z'_{i-1} \leq z_{i-1} = z_i$.  Otherwise, observe that
$y'_i \leq y_i$ and therefore $z'_i \leq z_i$.
\qed
\end{proof}


\subsection{Variable Radii}

We present an algorithm that solves the constrained \bcvr decision
problem.


Before presenting our algorithm, we need a few definitions.
Given a \bcvr instance $(x,b)$ and $t>0$, if sensor $i$ moves from
$x_i$ to $p \in [l(i),u(i)]$, then we may assume without loss of
generality that its radius is $r_i(p,t) = \sqrt[\alpha]{(b_i -
a|p-x_i|)/t}$.

%
\begin{wrapfigure}{r}{2.15in}  
\centering
\begin{tikzpicture}[scale=3.5,fill=lightgray]
  \draw[->] (-0.75,0) -- (0.75,0) node[below]{$d$};
  \draw (0.5,-0.02) -- (0.5,0.02) node[above] {$\frac{b_i}{a}$};
  \draw (0.375,-0.02) node[below] {$d^t_i$} -- (0.375,0.02);
  \draw (-0.5,-0.02) node[below] {$-\frac{b_i}{a}$} -- (-0.5,0.02);
  \draw (-0.375,-0.02) -- (-0.375,0.02) node[above] {$-d^t_i$};
  \draw[->] (0,-0.75) -- (0,0.75); 
  \draw (-0.02,0.6125) node[left] {$g^t_i(d^t_i)$} -- (0.02,0.6125);
  \draw (-0.02,0.5) node[left] {$\sqrt[\alpha]{b_i/t}$} -- (0.02,0.5);
  \draw (-0.02,-0.6125) -- (0.02,-0.6125) node[right] {$h^t_i(-d^t_i)$};
  \draw (-0.02,-0.5) -- (0.02,-0.5) node[right] {$-\sqrt[\alpha]{b_i/t}$};
  \draw[blue,thick,domain=-0.5:0.5] 
    plot (\x, { \x + sqrt( (1 - 2*abs(\x))/4 ) } );
  \draw[red,thick,domain=-0.5:0.5] 
    plot (\x, { \x - sqrt( (1 - 2*abs(\x))/4 ) } );
  \draw[dashed,domain=-0.5:0.5] plot (\x, \x );
  \draw[dotted,<->] (0.2,0.59) -- (0.2,-0.19);
\end{tikzpicture}
\vspace{-5pt}
\label{fig:functions}
\caption{Depiction of the functions $g_i^t(d)$ and $h_i^t(d)$ for 
$a=2$, $\alpha=2$, $b_i=1$, and $t=4$.  The top (blue) curve
corresponds to $g_i^t(d)$, and the bottom (red) curve corresponds to
$h_i^t(d)$.  The dashed line corresponds to the location of sensor
$i$, while the vertical interval between the curves is the interval
that is covered by $i$ at distance $d$ from $x_i$.}
\vspace{-20pt}
\end{wrapfigure}
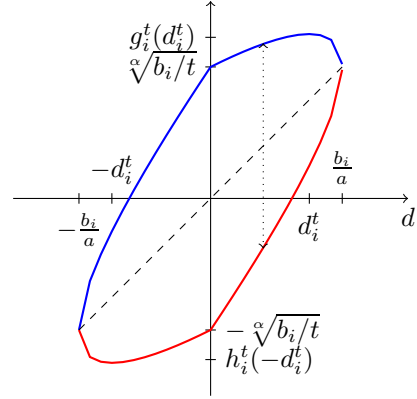
Similarly to Algorithm~\fixed, our algorithm tries to cover $[0,1]$ by
deploying sensors one by one, such that the length of the covered
prefix $[0,z]$ is maximized.  This motivates the following
definitions.
Let $d \in [-\frac{b_i}{a},\frac{b_i}{a}]$ denote the distance
traveled by sensor $i$, where $d >0$ means traveling right, and $d<0$ means
traveling left.  If a sensor travels a distance $d$, then its lifetime $t$
sustaining radius is given by $\sqrt[\alpha]{(b_i-a|d|)/t}$.  Given
$t$, we define:
\[
g^t_i(d) \eqdf d + \sqrt[\alpha]{(b_i-a|d|)/t}
~.
\]
$g^t_i(d)$ is the \emph{right reach} of sensor $i$ at distance $d$
from $x_i$, i.e., the rightmost point that $i$ covers when it has
traveled a distance of $d$ and the required lifetime is $t$.
Similarly define $h^t_i(d) \eqdf g^t_i(-d)$ is the \emph{left reach}
of sensor $i$ at distance $d$ from $x_i$.  See depiction in
Figure~\ref{fig:functions}.
%
%
We explore these functions in the next lemma whose proof is given in
the appendix.

\begin{lemma}
\label{lemma:maxreach}
Let $t>0$.  For any 
$i$, the distance $d^t_i$ maximizes $g^t_i(d)$, where
\begin{align*}
d^t_i & = 
\begin{cases}
\frac{b_i}{a} - \inv{\alpha} \sqrt[\alpha-1]{\frac{a}{\alpha t}}
              & \alpha>1 \\
\frac{b_i}{a} & \alpha=1, a < t \\
0            & \alpha=1, a \geq t
\end{cases}
&
g^t_i(d^t_i) & = 
\begin{cases}
\frac{b_i}{a} +\paren{1- \inv{\alpha}}
\sqrt[\alpha-1]{\frac{a}{\alpha t}}    & \alpha>1 \\
\frac{b_i}{\min\set{a,t}}              & \alpha=1
\end{cases}
\end{align*}
%
If $\alpha>1$ or $a \neq t$, $g^t_i$ is increasing for
$d<d^t_i$, and decreasing for $d>d^t_i$.
If $\alpha=1$ and $a=t$, $g^t_i$ is constant, for $d \geq 0$, and it
is increasing for $d<0$.
\end{lemma}


Given a point $z \in [0,1]$, the \emph{attaching position} of sensor
$i$ to $z$, denoted by $p_i(z,t)$, is the position $p$ for which
$p-r_i(p,t)=z$ such that $p+r_i(p,t)$ is maximized, if such a position
exist.  If such a point does not exist we define $p_i(z,t) = \infty$.
Observe that by Lemma~\ref{lemma:maxreach} there may be at most two
points that satisfy the equation $p-r_i(p,t)=z$.
Such a position can either be found explicitly or numerically as it
involves solving an equation of degree $\alpha$.  
We ignore calculation inaccuracies for ease of presentation.  These
inaccuracies are subsumed by the additive factor.

Algorithm~\variable is our decision algorithm for \bcvr.  It first
computes $u$ and $l$.  If there is a sensor $i$, such that $l(i) >
u(i)$, it outputs NO.  Then, it deploys the sensors one by one
according to $\prec$ with the goal of extending the coverage interval
$[0,z]$.
If $i$ cannot increase the covering interval it is placed at
$\max\{l(i),y_{i-1}\}$ so as not to block sensor $i+1$.
If $i$ can increase coverage, it is placed in $[l(u),u(i)]$ such that
$z$ is covered and coverage to the right is maximized.
It may be the case that the best place for $i$ is to the left of
previously positioned sensors.  In this case the algorithm moves the
sensors such that coverage and order are maintained.
Finally, if $z<1$ after placing sensor $n$, the algorithm outputs NO,
and otherwise it outputs YES.

\begin{algorithm}[t]
\begin{small}
\begin{algorithmic}[1]
\caption{: \variable$(x,b,t)$}
\label{Alg:decide-variable}
\State Compute $l$ and $u$
\State \textbf{if} there exists $i$ such that $u(i)<l(i)$
       \textbf{then} \Return{NO}
\State $z \gets 0$
\For{$i=1 \to n$}
   \State $q_L(i) \gets \min\set{ \max\set{x_i - d^t_i, l(i)}, u(i)}$ 
   \State $q_R(i) \gets \max\set{ \min\set{x_i + d^t_i, u(i)}, l(i)}$ 
   \If{$z \not\in [q_L(i) - r_i(q_L(i),t),q_R(i) + r_i(q_R(i),t)]$}
      \State $y_i \gets\max\set{l(i),y_{i-1}}$
             and $r_i \gets 0$
\Comment{$y_0 = 0$}
   \Else
      \State $y_i \gets \max\set{\min\set{p_i(z,t),u(i),x_i + d_i^t},l(i)}$
             and $r_i \gets r_i(y_i,t)$
      \State $S \gets \set{k : k < i, y_i < y_k}$
      \State $y_k \gets y_i$ and $r_k \gets 0$, for every $k \in S$
%
%
      \State $z \gets y_i + r_i$
   \EndIf
\EndFor
\State \textbf{if} $z<1$ \textbf{then} \Return{NO}
\State \textbf{else} \Return{YES}
\end{algorithmic}
\end{small}
\end{algorithm}

$l$ and $u$ can be computed in $O(n)$ time.  There are $n$ iterations
of the main loop, each taking $O(n)$ time (assuming that computing
$p_i(z,t)$ takes $O(1)$ time), thus the running time of the algorithm
is $O(n^2)$.

We now prove the correctness of Algorithm~\variable.
We define 
\[
P(i)
= \set{p : p \in [l(i),u(i)] \text{ and } z \in [p-r_i(p,t),p+r_i(p,t)]}
~. 
\]
$P(i)$ is the set of points from which sensor $i$ can cover $z$.
Observe that $P(i)$ is an interval due to Lemma~\ref{lemma:maxreach}.
Hence, we write $P(i) = [p_L(i),p_R(i)]$.

In the next two lemmas it is shown that when the algorithm checks
whether $z \not\in [q_L(i) - r_i(q_L(i),t), q_R(i) + r_i(q_R(i),t)]$
it actually checks whether $P(i) = \emptyset$, and that $y^*_i \eqdf
\max \set{\min \set{p_i(z,t),u(i),x_i + d_i^t},l(i)}$ is equal to
$\argmax_{p \in P} \set{p+r_i(p,t)}$.  Hence, in each iteration we
check whether $[0,z]$ can be extended, and if it can, we take the best
possible extension.

\begin{lemma}
\label{lemma:useless}
$[p_L(i),p_R(i)] \subseteq [q_L(i),q_R(i)]$.  Moreover, $P(i) =
\emptyset$ if and only if $z \not\in [q_L(i) - r_i(q_L(i),t), q_R(i) +
r_i(q_R(i),t)]$.
\end{lemma}
\begin{proof}
By Lemma~\ref{lemma:maxreach} $q_L(i)$ is the location that maximized
coverage to the left, and $q_R(i)$ is the location that maximized
coverage to the right.
\qed
\end{proof}

\begin{lemma}
\label{lemma:usefull}
If $P(i) \neq \emptyset$, then $y^*_i = \argmax_{p \in P(i)}
\set{p+r_i(p,t)}$.
\end{lemma}
\begin{proof}
By Lemma~\ref{lemma:maxreach}, there are three cases:
\begin{itemize}
\item If $x_i + d^t_i \in P(i)$, then 
      $\argmax_{p \in P(i)} \set{p+r_i(p,t)} = x_i + d^t_i$.

      $y^*_i = x_i + d_i^t$, since $p_i(z,t) \geq x_i + d^t_i$.

\item If $x_i + d^t_i > p_R(i)$, then 
      $\argmax_{p \in P(i)} \set{p+r_i(p,t)} = p_R(i)$.  

      $y^*_i = \min\set{p_i(z,t),u(i)}$, since $p_R(i) =
      \min\set{p_i(z,t),u(i)} \geq l(i)$.

\item If $x_i + d^t_i < p_L(i)$, then 
      $\argmax_{p \in P(i)} \set{p+r_i(p,t)} = p_L(i)$.  

      $y^*_i = l(i)$, since $q_L(i) = l(i) > x_i + d_i^t \geq
      \min\set{p_i(z,t),u(i),x_i + d_i^t}$. \qed
\end{itemize}
\end{proof}

We are now ready to prove the correctness of our algorithm.

\begin{theorem}
\label{thm:search-bcvr}
\variable solves the constrained \bcvr decision problem.
\end{theorem}


\subsection{Parametric Search Algorithms}

We design parametric search algorithms for constrained \bcfr and \bcvr.

Since we have an algorithm that, given $t$ and an order $\prec$,
decides whether there exists a solution that satisfies $\prec$ with
lifetime $t$, we can perform a binary search on $t$.
The maximum lifetime of a given instance is bounded by the lifetime of
this instance in the case where $a=0$.  In Appendix~\ref{sec:friction}
we show that the lifetime in the fixed case is at most $\max_i
\set{b_i/\rho_i^\alpha}$ and that in the variable radii case it is at most
$(2\sum_j \sqrt[\alpha]{b_j})^\alpha$.
Hence, the running time of the parametric search in polynomial in the
input size and in the $\log \inv{\eps}$, where $\eps$ is the accuracy
parameter.


%% file: zero.tex
\section{Sensors are Located on the Edges of the Barrier}
\label{sec:zero}

In this section we consider the special case in which the initial
locations are on either edge of the barrier, namely the case where $x
\in \set{0,1}^n$.
For both \bcvr and \bcfr we show that, given an achievable lifetime
$t$, there exists a solution with lifetime $t$ in which the sensors
satisfy a certain ordering.  
In the case of \bcvr, the ordering depends only on the battery sizes,
and hence we may use the parametric search algorithm for constrained
\bcvr from Section~\ref{sec:search} to solve this special case of \bcvr.
In the case of \bcfr, the ordering depends on $t$, and therefore may
change.  Even so, we may use parametric search for this special case of
\bcfr since, given $t$, the ordering can be computed in polynomial
time.


\paragraph*{\bf Fixed radii.}
We start by considering the special case of \bcfr in which all sensors
are located at $x=0$.  (The case where $x=1$ is symmetric.)

Given a \bcfr instance $(0,b,\rho)$ and a lifetime $t$, the
\emph{maximum reach} of sensor $i$ is defined as the farthest point
from its initial position that sensors $i$ can cover while maintaining
lifetime $t$, and is given by: $f_t(i) = (b_i - t \rho_i^\alpha)/a +
\rho_i$, if $t \rho_i^\alpha \leq b_i$, and $f_t(i) = 0$, otherwise.
We assume without loss of generality in the following that the sensors
are ordered according to \emph{reach ordering}, namely that $i<j$ if
and only if $f_t(i) < f_t(j)$.  Also, we ignore sensors with zero
reach, since they must power down.  Hence, if $f_t(i)=0$, we place $i$
at 0 and set its radius to 0.
Let $t$ be an achievable lifetime, we show that there exists a
solution $(y,r)$ with lifetime $t$ such that sensors are deployed
according to reach ordering.

\begin{lemma}
\label{lemma:zero-bcfr}
Let $(0,b,\rho)$ be a \bcfr instance and let $p \in (0,1]$.  Suppose
that there exists a solution that covers $[0,p]$ for $t$ time.  Then,
there exists a solution that covers $[0,p]$ lifetime for $t$ time that
satisfies reach ordering.
\end{lemma}


\paragraph*{\bf Variable radii.}
We now consider the case where $x=0$ for \bcvr.  ($x=1$ is symmetric.)

Given a \bcvr instance $(0,b)$ and a lifetime $t$, the \emph{maximum
reach} of sensor $i$ is $g^t_i(d^t_i)$.  Note that if the sensors are
ordered by battery size, namely that $i<j$ if and only if $b_i < b_j$,
they are also ordered by reach.  Thus, we assume in the following that
sensors are ordered by battery size. 
%
Let $t$ be an achievable lifetime.  We show that there exists a
deployment $y$ with lifetime $t$ such that sensors are deployed
according to the battery ordering, namely $b_i \leq b_j$ if and only
if $y_i \leq y_j$.

We need the following technical lemma.

\begin{lemma}
\label{lemma:tech}
Let $c_1,c_2,d_1,d_2 \geq 0$ such that 
\begin{inparaenum}[(i)]
\item $d_1 < c_1 \leq c_2 < d_2$, and 
\item $c_1+c_2 > d_1+d_2$.
\end{inparaenum}  
Also let $\alpha \geq 1$.  Then, $\sqrt[\alpha]{c_1} +
\sqrt[\alpha]{c_2} > \sqrt[\alpha]{d_1} + \sqrt[\alpha]{d_2}$
\end{lemma}

\begin{lemma}
\label{lemma:zero-bcvr}
Let $(0,b)$ be a \bcvr instance and let $p \in (0,1]$.  Suppose that
there exists a deployment that covers $[0,p]$ for $t$ time.  Then,
there exists a deployment that covers $[0,p]$ lifetime for $t$ time
that satisfies battery ordering.
\end{lemma}
\begin{proof}
Given a solution that covers $[0,p]$ with lifetime $t$, a pair of
sensors is said to violate battery ordering if $b_i < b_j$ and $y_i >
y_j$.
Let $y$ be a solution with lifetime $t$ for $(x,b,r)$ that minimizes
battery ordering violations.  If there are no violations, then we are
done.  Otherwise, we show that the number of violations can be
decreased.

If $y$ has ordering violations, then there must exist at least one
violation due to a pair of adjacent sensors.  Let $i$ and $j$ be such
sensors.  
If the barrier is covered without $i$, then $i$ is moved to $y_j$.
(Namely $y'_k = y_k$, for every $k \neq i$, and $y'_i = y_j$.)  $y'$
is feasible, since $i$ moves to the left.  Otherwise, if the barrier
is covered without $j$, then $j$ is moved to $y_i$ and $j$'s radius is
decreased accordingly.

Otherwise, both sensors actively participate in covering the barrier,
which means that the interval $[y_j-r_j,y_i+r_i]$ is covered by $i$
and $j$.  In this case, we place $i$ at $y'_i$ with radii $r'_i$, such
that $y'_i - r'_i = y_j - r_j$.  We place $j$ at the rightmost
location $y'_j$ such that $y'_j \leq y_i$ and $y'_j - r'_j \leq y'_i +
r'_i$.  If $y'_j = y_i$ then we are done, as sensor $j$ has more
battery power at $y_i$ than $i$ does at $y_i$.  Otherwise, we may
assume that $y'_j - r'_j= y'_i +r'_i$. We show that it must be that
$y'_j + r'_j \geq y_i+r_i$.
We have that $y'_i < y_j$ and $y'_j< y_i$.  It follows that
$\beta'_i + \beta'_j > \beta_i + \beta_j$, where $\beta_i = b_i - a
y_i$.  Also, notice that $\beta_i < \beta'_j < \beta_j$ and $\beta_i <
\beta'_i < \beta_j$.  It follows that
\(
r'_i + r'_j
= \sqrt[\alpha]{\beta'_i/t} + \sqrt[\alpha]{\beta'_j/t}
> \sqrt[\alpha]{\beta_i/t} + \sqrt[\alpha]{\beta_j/t}
= r_i + r_j
\),
where the inequality is due to Lemma~\ref{lemma:tech}.
Hence, 
\(
y'_j + r'_j 
=    (y_j - r_j) +  2r'_i + 2r'_j
>    (y_j - r_j) +  2r_i + 2r_j
\geq y_i+r_i
\).

Since $i$ moves to the left, it may bypass several sensors.  In this
case we move all sensors with smaller batteries that were bypassed by
$i$, to $y'_i$.  This can be done since these sensors are not needed
for covering to the right of $y'_i-r'_i$.
Similarly, since $j$ moves to the right, it may bypass several
sensors.  As long as there is a sensor with larger reach that was
bypassed by $j$, let $k$ be the rightmost such sensor.  Notice that
$k$ is not needed for covering to the left of $y'_j$.  Hence, if $y_k
+ r_k \geq y'_j + r'_j$, we move $j$ to $y_k$.  Otherwise, we move $k$
to $y'_j$.

In all cases, we get a deployment $y'$ that covers $[0,p]$ with
lifetime $t$ with a smaller number of violations than $y$.  A
contradiction.
\qed
\end{proof}


\paragraph*{\bf Separation.}
We are now ready to tackle the case where $x \in \set{0,1}^n$.

We start with the fixed radii case.  Given a \bcfr instance $(x,b,r)$
and a lifetime $t$, we assume without loss of generality that the
sensors are ordered according to the following \emph{bi-directional
reach} order: first the sensors that are located at 0 according to
reach order, and then the sensors that are located at 1 according to
reverse reach order.

We show that we may assume that the sensors are deployed using the
bi-directional reach order.  The first step is to show that 
the sensors that are located at 0 are deployed to the left
of the sensors that are placed at 1.  

\begin{lemma}
\label{lemma:sep-bcfr}
Let $(x,b,\rho)$ be a \bcfr instance, where $x \in \set{0,1}^n$, and
let $t$ be an achievable lifetime.  Then, there exists a feasible
solution $(y,r)$ with lifetime $t$ such that $y_i \leq y_j$, for every
$i \leq \ell < j$.
\end{lemma}

Next we show that we may assume that the sensors are deployed using
the bi-directional reach order.  

\begin{theorem}
\label{thm:bi-bcfr}
Let $(x,b,\rho)$ be a \bcfr instance, and let $t$ be an achievable
lifetime.  Then there exists a feasible solution $(y,r)$ with lifetime
$t$ such that the sensors are deployed using bi-directional reach
order.
\end{theorem}

We treat the variable radius case in a similar manner.  Given a \bcvr
instance $(x,b)$, we assume without loss of generality that the
sensors are ordered according to a \emph{bi-directional battery}
order: first the sensors that are located at 0 according to battery
order, and then the sensors that are located at 1 according to reverse
battery order.
The proofs of the next lemma and theorem are similar to the proofs of
Lemma~\ref{lemma:sep-bcfr} and Theorem~\ref{thm:bi-bcfr}.

\begin{lemma}
Let $(x,b)$ be a \bcvr instance, where $x \in \set{0,1}^n$, and let
$t$ be an achievable lifetime.  Then, there exists a feasible solution
$(y,r)$ with lifetime $t$ such that $y_i \leq y_j$, for every $i \leq
\ell < j$.
\end{lemma}

\begin{theorem}
Let $(x,b)$ be a \bcvr instance, and let $t$ be an achievable
lifetime.  Then there exists a feasible solution $(y,r)$ with lifetime
$t$ such that the sensors are deployed using bi-directional battery
order.
\end{theorem}


%% file: conclusion.tex
\section{Discussion and Open Problems}

We briefly discuss some research directions and open problems.
%
We have shown that 
\bcvr is strongly NP-Hard.  Finding an approximation algorithm or 
showing hardness of approximation remains open.
In a natural extension model, sensors could be located anywhere in the
plane and asked to cover a boundary or a circular boundary.  In a more
general model the sensors need to cover the plane or part of the plane
where their initial locations could be anywhere.
Another model which can be considered is the {\em duty cycling} model
in which sensors are partitioned into shifts that cover the barrier.
Bar-Noy \textit{et al.}~\cite{BBR12} considered this model for
stationary sensors and $\alpha=1$.  Extending it to moving sensors and
$\alpha > 1$ is an interesting research direction.
Finally, in the most general covering problem with the goal of
maximizing the coverage lifetime, 
sensors could change their locations
and sensing ranges at any time.  
Coverage terminates when all the batteries are drained.

%% file: friction.tex
\section{Extreme Movement Costs}
\label{sec:friction}

In this section we consider the two extreme cases, the static case ($a
= \infty$) and the fully dynamic case ($a=0$).


\subsection{The Static Case}

In the static case the initial deployment is the final deployment,
i.e., $y = x$, and therefore a feasible solution is a radii assignment
$r$, such that $[0,1] \subseteq \cup_i [x_i - r_i, x_i + r_i]$.

We describe a simple algorithm for static \bcfr.  First, if $[0,1]
\not\subseteq \cup_i [x_i - \rho_i, x_i + \rho_i]$, then the maximum
lifetime is 0.  Otherwise, compute $t_i = b_i/\rho_i^\alpha$ for every
$i$, and let $S = \emptyset$.  Then, as long as $S$ does not cover the
barrier, add $i = \argmax_{i \not\in S} t_i$ to $S$.  Finally, assign
$r_i = \rho_i$, for $i \in S$, and $r_i = 0$, for $i \not\in S$.  The
correctness of this algorithm is straightforward.

Bar-Noy \textit{et al}.~\cite{BBR12corr} presented a polynomial time
algorithm for static \bcvr with $\alpha = 1$.  This algorithm readily
extends to static \bcvr with $\alpha>1$.  We refer the reader
to~\cite{BBR12corr} for the details.


\subsection{Fully Dynamic Case}

In the fully dynamic case movement is for free, and therefore any
radii vector $r$, such that $\sum_i 2r_i \geq 1$, has a deployment
vector $y$ such that $(y,r)$ is a feasible pair.  (e.g., $y_i =
\sum_{j=1}^{i-1} 2r_j + r_i$, for every $i$.)

We describe a simple algorithm for fully dynamic \bcfr.  First, if $\sum_i
2\rho_i < 1$, the maximum lifetime is 0.  Otherwise, compute $t_i =
b_i/\rho_i^\alpha$ for every $i$, and let $S = \emptyset$.  Then, as
long as $\sum_{i \in S} 2\rho_i < 1$, add $i = \argmax_{i \not\in S}
t_i$ to $S$.  Finally, assign $r_i = \rho_i$, for $i \in S$, and $r_i
= 0$, for $i \not\in S$.  The correctness of this algorithm is
straightforward.

We now consider fully dynamic \bcvr.  Given a feasible radii vector
$r$ and a corresponding deployment vector $y$, the lifetime of sensors
$i$ is simply $L_i(y,r) = b_i/r_i^\alpha$, and the lifetime of the
system is $L(y,r) = \min_i L_i(y,r)$.

\begin{theorem}
\label{thm:nofriction}
Let $a=0$.  Given a \bcvr instance, the radii assignment 
\(
r_i = \frac{\sqrt[\alpha]{b_i}}{2\sum_j \sqrt[\alpha]{b_j}}
\), 
for every $i$, is optimal.
\end{theorem}
\begin{proof}
First, observe that
\(
2 \sum_i r_i 
= \sum_i \frac{\sqrt[\alpha]{b_i}}{\sum_j \sqrt[\alpha]{b_j}}
= 1
\),
which means that $r$ is feasible. 
Furthermore,
\[
L_i(r)
= b_i/r_i^\alpha
= b_i \cdot 
  \paren{\frac{2\sum_j \sqrt[\alpha]{b_j}}{\sqrt[\alpha]{b_i}}}^\alpha
= \paren{2\sum_j \sqrt[\alpha]{b_j}}^\alpha
~,
\]
for every $i$.  Hence, $L(r) = (2\sum_j \sqrt[\alpha]{b_j})^\alpha$.

We show that $L(r) < L(r')$, for any radii assignment $r' \neq r$.
Since $r'$ is feasible, we have that $2 \sum_i r'_i \geq 1$.  It follows
that there exists $i$ such that $r'_i > r_i$.  Hence, $L(r') \geq
L_i(r') > L_i(r) = L(r)$.
\qed
\end{proof}

%% file: proofs.tex
\section{Omitted Proofs}

\begin{proof}[of Lemma~\ref{lemma:maxreach}]
First consider the case where $\alpha>1$.
For $d \in [b_i/a,0)$ we get
\[
\frac{\partial h^t_i}{\partial d}
= 1 + \frac{a}{\alpha t}\paren{\frac{b_i+ad}{t}}^{1/\alpha-1}
> 0
~.
\]
For $d \in (0,b_i/a]$, the derivative of $g^t_i$ is given by
\[
\frac{\partial g^t_i}{\partial d}
= 1 - \frac{a}{\alpha t}\paren{\frac{b_i-ad}{t}}^{1/\alpha-1}
~.
\]
It follows that $\frac{\partial g^t_i}{\partial d}(d)=0$ when 
\[
d=d^t_i 
= \frac{b_i}{a} - \frac{t}{a} \paren{\frac{a}{\alpha t}}^{\alpha/(\alpha-1)} 
= \frac{b_i}{a} - \inv{\alpha} \sqrt[\alpha-1]{\frac{a}{\alpha t}}
~.
\]
Furthermore, $g^t_i(d)>0$ when $d<d^t_i$, and $g^t_i(d)<0$ when
$d>d^t_i$.  The radius at this distance is
$\sqrt[\alpha-1]{\frac{a}{\alpha t}}$.  The maximum reach is thus
\[
g^t_i(d^t_i) 
= \frac{b_i}{a} +\paren{1- \inv{\alpha}} \sqrt[\alpha-1]{\frac{a}{\alpha t}}
~.
\]

For $\alpha=1$ we have
\[
g^t_i(d) = 
\begin{cases}
d(1-a/t) + b_i/t & d \geq 0, \\
d(1+a/t) + b_i/t & d < 0. 
\end{cases}
\]
Hence, 
\[
\frac{\partial g^t_i(d)}{\partial d} = 
\begin{cases}
1-a/t & d > 0, \\
1+a/t & d < 0. 
\end{cases}
\]
If $d>0$, we have several cases.
If $a>t$, the maximum occurs at $d^t_i=0$ and $g^t_i(d^t_i) =
\frac{b_i}{t}$.  If $a=t$, $g^t_i(d)=\frac{b_i}{t}$, for any $d \leq
\frac{b_i}{a}$.  If $a<t$, the function is increasing for any $d \leq
\frac{b_i}{a}$, and thus $d^t_i = \frac{b_i}{a}$ and $g^t_i(d^t_i) =
\frac{b_i}{a}$.  Hence, $g^t_i(d^t_i) =\frac{b_i}{\min\set{a,t}}$.
\qed
\end{proof}

\medskip

\begin{proof}[Proof of Theorem~\ref{thm:search-bcvr}]
If $u(i)<l(i)$ for some $i$, then no deployment that satisfies the
order $\prec$ exists by Observation~\ref{obs:lu}.  Hence, the
algorithm responds correctly.

We show that if the algorithm outputs YES, then the computed solution
is feasible.  First, notice that $y_{i-1} \leq y_i$, for every $i$, by
construction.  We prove by induction on $i$, that $y_j \in
[l(j),u(j)]$ for every $j \leq i$.
%
Consider the $i$th iteration.  If $z \not\in [q_L(i) - r_i(q_L(i),t),
q_R(i) + r_i(q_R(i),t)]$, then $y_i \in [l(i),u(i)]$, since
$\max\set{l(i),y_{i-1}} \leq \max\set{u(i),u(i-1)} \leq u(i)$.
Otherwise, $y_i = \max\set{\min\set{p_i(z,t),u(i),x_i + d_i^t},l(i)}
\in [l(i),u(i)]$.  Furthermore, if $j < i$ is moved to the left due
$i$, then $y_j = y_i \geq l(i) \geq l(j)$.
Finally, let $z_i$ denote the value of $z$ after the $i$th iteration.
(Initially, $z_0 = 0$.)  We prove by induction on $i$ that $[0,z_i]$
is covered.  Consider iteration $i$.  If $r_i=0$, then we are done.
Otherwise, $z_{i-1} \in [y_i-r_i,y_i+r_i]$ and $z_i = y_i+r_i$, and
the sensors in $S$ can be powered down and moved, since
$[y_j-r_j,y_j+r_j] \subseteq [y_i-r_i,y_i+r_i]$, for every $j \in S$.

Finally, we show that if the algorithm outputs NO, there is no
feasible solution.
We prove by induction that $[0,z_i$] is the longest interval that can
be covered by sensors $1,\ldots,i$.
In the base case, observe that $z_0=0$ is optimal.  
For the induction step, let $y'$ be a deployment of $1,\ldots,i$ that
covers the interval $[0,z'_i]$.  Let $[0,z'_{i-1}]$ be the interval
that it covers by $1,\ldots,i-1$.  By the inductive hypothesis,
$z'_{i-1} \leq z_{i-1}$.  
If $z'_i \leq z_{i-1}$, then we are done.  Otherwise, we have that
$y'_i + r_i(y'_i,t) > z_{i-1}$.  In this case we have that $y'_i
\in P(i)$.  It follows, by Lemma~\ref{lemma:useless}, that 
we place $i$ at $y_i = y^*_i$.  By Lemma~\ref{lemma:usefull} we have
$y_i$ is better than $y'_i$ in terms of coverage to the right, namely
$z_i= y_i + r_i(y_i,t) \geq y'_i + r_i(y'_i,t) = z'_i$.
%
\qed
\end{proof}

\medskip

\begin{proof}[Proof of Lemma~\ref{lemma:zero-bcfr}]
We first prove that we may focus on feasible solutions where $r=\rho$.
Given a feasible solution $(y,r)$ that covers $[0,p]$ with lifetime
$t$, we define $y'_i = y_i$, if $r_i = \rho_i$, and $y'_i = 0$,
otherwise.  The pair $(y',\rho)$ clearly covers $[0,p]$ with lifetime
$t$.  (Recall that we ignore sensors with zero reach.)

Given a solution that covers $[0,p]$ with lifetime $t$, a pair of
sensors is said to violate reach ordering if $i < j$ and $y_i > y_j$.
Let $(y,\rho)$ be a solution with lifetime $t$ for $(0,b,\rho)$ that
minimizes reach ordering violations.  If there are no violations, then
we are done.  Otherwise, we show that the number of violations can be
decreased.

If $y$ has ordering violations, then there must exist at least one
violation due to a pair of adjacent sensors.  Let $i$ and $j$ be such
sensors.  
If the barrier is covered without $i$, then $i$ is moved to $y_j$.
(Namely $y'_k = y_k$, for every $k \neq i$, and $y'_i = y_j$.)  $y'$
is feasible, since $i$ moves to the left.
Otherwise, if the barrier is covered without $j$, then $j$ is moved to
$y'_j = \min\set{y_i,f_t(j)-\rho_j}$.  If $y'_j = y_i$, then we are
done.  If $y'_j < y_i$, then $[y_i-\rho_i,y_i+\rho_i] \subseteq
[y_j-\rho_j,y_j+\rho_j]$, since $f_t(j) > f_t(i)$.  It follows that
the barrier is covered without $i$, and so we can move $i$ to $y'_j$.
Since $y'_j \leq f_t(j)-\rho_j$, and $i$ moves to the left, we get a
feasible deployment.

If both sensors participate in the cover, we define a new deployment
$y'$ by moving $i$ to $y'_i = y_j + (\rho_i - \rho_j)$ and moving $j$
to $y'_j = y_i + (\rho_i - \rho_j)$.  The interval $[0,p]$ is covered,
since $[y_j - \rho_j, y_i + \rho_i]$ is covered.  Also, $y'_i \leq
y'_j$.  Furthermore, $i$ and $j$ can maintain their radii for $t$
time, since $y'_i \leq y_i$ and $f_t(j) > f_t(i)$.
Since $i$ moves to the left, it may bypass several sensors.  In this case
we move all sensors with smaller reach that were bypassed by $i$, 
to $y'_i$.
Since $j$ moves to the right, it may bypass several sensors.  As long
as there is a sensor with larger reach that was bypassed by $j$, let
$k$ be the rightmost such sensor, and move both $j$ and $k$ to
$\min\set{y'_j,f_t(k)-\rho_k}$.  Notice that $k$ is not needed for
covering to the left of $y'_j$, and thus it can be moved to the right,
as long as it has the power to do so.  If $k$ cannot move to $y'_j$,
it follows that $j$ is not needed for covering to the right of $y'_k$.

In all cases, we get a deployment $y'$ that covers $[0,p]$ with
lifetime $t$ with a smaller number of violations than $y$.  A
contradiction.
\qed
\end{proof}

\medskip

\begin{proof}[of Lemma~\ref{lemma:tech}]
The case where $\alpha=1$ is immediate, so henceforth we assume that
$\alpha>1$.
Let $s = c_1+c_2$, and let $d_2' = s-d_1$.  We prove that $d_1^{1/\alpha} +
(s-d_1)^{1/\alpha} < c_1^{1/\alpha} + (s-c_1)^{1/\alpha}$.  Since $d'_2 >
d_2$ the lemma follows.

Define $f(x) = x^{1/\alpha} + (s-x)^{1/\alpha}$.  
The derivative is:
\[
f'(x) = \frac{x^{1/\alpha - 1}}{\alpha} - \frac{(s-x)^{1/\alpha - 1}}{\alpha}
      = \inv{\alpha x^{1 - 1/\alpha}} - \inv{\alpha (s-x)^{1 - 1/\alpha}}
~.
\]
$f'(x) = 0$ implies that $x = \frac{s}{2}$ and $f'(x)>0$ for $0 \leq x
< \frac{s}{2}$.  It follows that $f(x)$ is an increasing function in
the interval $(0,\frac{s}{2})$.  Thus we have $f(c_1) > f(d_1)$.
\qed
\end{proof}

\medskip

\begin{proof}[of Lemma~\ref{lemma:sep-bcfr}]
Given a deployment $y$ for $(x,b,r)$, a pair of sensors is called
\emph{bad} if $i \leq \ell < j$ and $y_i > y_j$.
Let $y$ be a deployment with lifetime $t$ for $(x,b,r)$ that minimizes
the number of bad pairs.  If there are no bad pairs, then we are done.
Otherwise, we show that the number of bad pairs can be decreased.
If $y$ has a bad pair, then there must exist at least one bad pair of
adjacent sensors.  Let $i$ and $j$ be such sensors.  We construct a
new deployment vector $y'$ as follows.  

If the barrier is covered without $i$, then $i$ is moved to $0$,
namely $y'_k = y_k$, for every $k \neq i$, and $y'_i = 0$.  Otherwise,
if the barrier is covered without $j$, then $j$ is moved to $1$,
namely $y'_k = y_k$, for every $k \neq j$, and $y'_i = 1$.  In both
cases the pair $(y',r)$ is feasible and has lifetime $t$.  Furthermore
the number of bad pairs decreases.  A contradiction.

If both $i$ and $j$ are essential to the cover, we define $y'$ as
follows:
\[
y'_k = 
\begin{cases}
y_j + (\rho_i - \rho_j) & k = i, \\
y_i + (\rho_i - \rho_j) & k = j, \\
y_k               & k \neq i, j. 
\end{cases}
\]
We show that $(y',r)$ is a feasible solution.  First, notice that
$y'_i = y_j + (\rho_i - \rho_j) < y_i$, since otherwise the barrier
can be covered without $j$.  Similarly, $y'_j = y_i + (\rho_i -
\rho_j) < y_j$.  Hence, $y'_k \leq y_k$, for $k \leq \ell$, and $y'_k
\geq y_k$, for $k > \ell$, which means that $y'$ consumes less power
than $y$.  Also the barrier is covered, since the interval $[y_j -
\rho_j, y_i + \rho_i]$ is covered by $i$ and $j$.  Finally, $y'_i =
y_j + (\rho_i - \rho_j) \leq y_i + (\rho_i - \rho_j) = y'_j$, and
therefore the number of bad pair decreases.  A contradiction.
\qed
\end{proof}

\medskip

\begin{proof}[of Theorem~\ref{thm:bi-bcfr}]
By Lemma~\ref{lemma:sep-bcfr} we know that there exists a deployment
$y$, such that $y_i \leq y_j$, for every $i \leq \ell < j$.  It
follows that sensors from 0 cover $[0,p_0]$ while sensors from 1 cover
$[p_1,1]$, where $p_0 \geq p_1$.
Lemma~\ref{lemma:zero-bcfr} implies that there is a deployment $y^0$
of the sensors from 0 that covers $[0,p_0]$ that satisfies reach
order, and that there is a deployment $y^1$ of sensors from 1 that
covers $[p_1,1]$ that satisfies reverse reach order.
Define 
\[
y'_i = 
\begin{cases}
y^0_i                    & i \leq \ell, \\
\max\set{y^1_i,y^0_\ell} & i > \ell.
\end{cases}
\]
$y'$ covers $[0,1]$ and it satisfies the bi-directional reach order.
\qed
\end{proof}

%% file: hardness.tex
\section{Hardness Results}
\label{sec:hardness}

In this section we show that 
\begin{inparaenum}[(i)]
\item \bcfr is NP-hard, even if $x \in p^n$, for any $p \in (0,1)$. 
\item There is no polynomial time multiplicative approximation 
      algorithm for \bcfr, unless P$=$NP, even if $x = p^n$.
\item There is no polynomial time algorithm that computes a solution 
      within an additive factor $\eps$, for some constant $\eps>0$,
      unless P$=$NP, even if $x = p^n$.
\item \bcvr is strongly NP-hard.
\end{inparaenum}
The hardness results apply to any $a > 0$ and $\alpha \geq 1$.

We note that throughout the section we assume that $\alpha$ is
integral for ease of presentation.  More specifically, we assume that
exponentiation with exponent $\alpha$ can be done in polynomial time.
Our constructions can be fixed by taking a numerical approximation
which is slightly larger than the required power.


\subsection{Fixed Radii}

The first result is obtained using a reduction from \partition.%
\footnote{A \partition instance consists of a list $a_1,\ldots,a_n$ of 
positive integers, and the goal is to decide whether there exists $I
\subseteq \set{1,\ldots,n}$ such that $\sum_{i \in I} a_i = \sum_{i
\not\in I} a_i$.}
Roughly speaking, our reduction uses a sensor that cannot move if it
is required to maintain its radius for one unit of time.  This sensor
splits the line into two segments, and therefore the question of
whether the given numbers can be partitioned into two subsets of equal
sum translates into the question of whether we can cover the two
segments for some time interval.  
%

\begin{lemma}
\label{lemma:bcfr-hard-half}
\bcfr is NP-hard, for any $a > 0$ and $\alpha \geq 1$, even if
$x = \half^n$.  
Furthermore, in this case it is NP-hard to decide whether the maximum
lifetime is zero or at least $a$.
\end{lemma}
\begin{proof}
Given a \partition instance $a_1,...,a_n$, let $B = \sum_i a_i$.  We
construct a \bcfr instance with $n+1$ sensors as follows: $x_i =
\half$, for every $i$; 
\begin{align*}
\rho_i = &
\begin{cases}
\frac{a_i}{2(B+1)} & i \leq n, \\
\frac{1}{2(B+1)}   & i = n+1; 
\end{cases}
&&&
b_i = &
\begin{cases}
a\rho_i^\alpha + \frac{a}{2} & i \neq n+1, \\
a\rho_i^\alpha               & i = n+1.
\end{cases}
\end{align*}
We show that $(a_1,\dots,a_n) \in \partition$ implies that there
exists a solution with lifetime $a$, and that the maximum lifetime is
zero if $(a_1,\dots,a_n) \not\in \partition$.

Suppose that $(a_1,\dots,a_n) \in \partition$, and let $I \subseteq
\set{1,\ldots,n}$ be such that $\sum_{i \in I} a_i = \half \sum_i a_i$.
Set $r_i = \rho_i$, for every sensor $i$.  Use sensor $n+1$ to cover
the interval $[\half-\inv{2B+2}, \half+\inv{2B+2}]$, the sensors that
correspond to $I$ to cover the interval $[0,\half-\inv{2B+2}]$, and
the rest of the sensors to cover the interval $[\half+\inv{2B+2},1]$.
This is possible, since $\sum_{i \in I} 2\rho_i = \half - \inv{2B+2}$,
and $\sum_{i \in \set{1,\ldots,n} \setminus I} 2\rho_i = \half -
\frac{1}{2B+2}$.  It is not hard to verify that a lifetime of $a$ is
achievable.

Suppose that $(a_1,\dots,a_n) \not\in \partition$, and assume that
there exists a solution $(y,r)$ with non-zero lifetime.  It must be
that $r_i = \rho_i$, for every $i$, since $\sum_i 2\rho_i = 1$.  Since
$\alpha \geq 1$, sensor $n+1$ cannot move more than $\inv{2B+2}$.  It
follows that $y_{n+1} = \half$, since all radii are multiples of
$\inv{2B+2}$.  Thus there is a subset $I \subseteq \set{1,\ldots,n}$
of sensors that covers $[0,\half-\inv{2B+2}]$, and
\(
\sum_{i \in I} a_i 
=    (B+1) \sum_{i \in I} 2\rho_i 
=    \half B  
\).
Hence, $(a_1,\dots,a_n) \in \partition$.  A contradiction.
\qed 
\end{proof}

The next step is to prove a similar result for any $p \in (0,1)$.
Since we already considered $p=\half$, we assume, without loss of
generality, that $p < \half$.

\begin{lemma}
\label{lemma:bcfr-hard-p}
\bcfr is NP-hard, for any $a > 0$ and $\alpha \geq 1$, even if $x = p^n$, 
where $p \in (0,\half)$.
Furthermore, in this case it is NP-hard to decide whether the maximum
lifetime is zero or at least $a$.
\end{lemma}
\begin{proof}
Given a \partition instance $a_1,...,a_n$, let $B = \sum_i a_i$.  
We construct a \bcfr instance with $n+3$ sensors as follows: $x_i =
p$, for every $i$; 
\begin{align*}
\rho_i = &
\begin{cases}
\frac{a_i d}{2(B+1)} & i \leq n, \\
\frac{d}{2(B+1)}     & i = n+1, \\
\frac{p-d/2}{2}      & i = n+2, \\
\frac{1-p-d/2}{2}    & i = n+3;
\end{cases}
&&&
b_i = &
\begin{cases}
a\rho_i^\alpha + a & i \neq n+1, \\
a\rho_i^\alpha     & i = n+1.
\end{cases}
\end{align*}
where $d = \min\set{p,1-2p}$.
We show that $(a_1,\dots,a_n) \in \partition$ implies that there
exists a solution with lifetime $a$, and that the maximum lifetime is
zero if $(a_1,\dots,a_n) \not\in \partition$.

Supposed that $(a_1,\dots,a_n) \in \partition$, and let $I \subseteq
\set{1,\ldots,n}$ such that $\sum_{i \in I} a_i = \sum_{i \not\in I}
a_i$.  Define $\bar{I} = \set{1,\ldots,n} \setminus I$.
Set $r_i = \rho_i$, for every $i$, and use the following deployment:
\begin{enumerate}
\item Sensor $n+1$ does not move and covers 
      $[p-\frac{d}{2B+2}, p+\frac{d}{2B+2}]$.
\item Sensor $n+2$ moves to $\frac{p-d/2}{2}$ and covers $[0,p-d/2]$.
\item Sensor $n+3$ moves to $\frac{1+p+d/2}{2}$ and covers 
      $[p+d/2,1]$. 
\item The sensors that correspond to $I$ deploy such that they 
      cover $[p-d/2,p-\frac{p}{2B+2}]$.
\item The sensors that correspond to $\bar{I}$ deploy such that they 
      cover $[p+\frac{p}{2B+2}, p+d/2]$.
\end{enumerate}
(See example in Figure~\ref{fig:partition}.)
This is possible, since 
\(
\sum_{i \in I} 2\rho_i 
= \sum_{i \in I} \frac{a_i d}{B+1}
= \frac{B d}{2(B+1)}
= \frac{d}{2} - \frac{d}{2B+2}
\),
and similarly $\sum_{i \in \bar{I}} 2r_i = \frac{d}{2}
-\frac{d}{2B+2}$.  It is not hard to verify that a lifetime of $a$ is
achievable.

Suppose that $(a_1,\dots,a_n) \not\in \partition$, and assume that
there exists a solution $(y,r)$ with non-zero lifetime.  Notice that
$\sum_i 2\rho_i = 1$, and thus it must be that $r_i = \rho_i$, for
every $i$.
Since $\alpha \geq 1$, the battery of sensor $n+1$ is depleted if it
moves a distance of $\frac{d}{2B+2}$.  This means that $y_{n+1} \in (p
- \frac{d}{2B+2}, p + \frac{d}{2B+2})$.  Since $y_{n+1} < p +
\frac{d}{2B+2} \leq p + \frac{d}{2} \leq p + (\half-p) = \half$, and
\(
\rho_{n+3} 
= \half - \frac{p}{2} - \min \set{\frac{p}{4},\inv{4}-\frac{p}{2}} 
= \max \set{ \half - \frac{3p}{4}, \inv{4}} 
\geq \inv{4}
\), 
it follows that $n+3$ must be deployed such that its covering interval
is to the right of the interval of $n+1$, namely $y_{n+3}- \rho_{n+3}
\geq y_{n+1} + \rho_{n+1}$.
Next, observe that $\rho_{n+2} + \rho_{n+3} = \frac{p-d/2}{2} +
\frac{1-p-d/2}{2} = \frac{1-d}{2}$.  Since $y_{n+1} + \rho_{n+1} > p -
\frac{d}{2B+2} + \frac{d}{2B+2} = p \geq d$, it follows that sensor
$n+2$ must be deployed such that its covering interval is to the left
of the interval of $n+1$, namely $y_{n+2}+ \rho_{n+2} \leq y_{n+1} -
\rho_{n+1}$.  Without loss of generality we assume that sensors
$n+2$ and $n+3$ are adjacent to 0 and 1, respectively.
Since all remaining radii are multiples of $\frac{d}{2B+2}$, it
follows that $y_{n+1} = p$.  Hence there is a subset $I \subseteq
\set{1,\ldots,n}$ of sensors that covers the remaining uncovered area
to the left of $p-\frac{d}{2B+2}$, while the rest of the sensors cover
the remaining uncovered area to the right of $p+\frac{d}{2B+2}$.  Thus
\[\textstyle
\sum_{i \in I} a_i 
= \frac{B+1}{d} \sum_{i \in I} 2\rho_i 
= \frac{B+1}{d} \paren{ \frac{d}{2} - \frac{d}{2B+2} }
= \half B
~.  
\]
Hence, $(a_1,\dots,a_n) \in \partition$.  A contradiction.
\qed 
\end{proof}

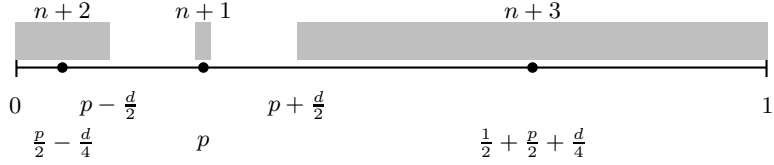
\begin{figure}[t]
\centering
\begin{tikzpicture}[scale=1,auto,thick,fill=lightgray,
dot/.style={circle,fill=black,inner sep=0pt,minimum size=4pt}] 
  \draw (0,0)[arrows=|-|] -- (10,0);
  \draw (0,-0.5) node {$0$};
  \draw (10,-0.5) node {$1$};
  \fill (2.4,0.1) rectangle (2.6,0.6);
  \node at (2.5,0) [dot] {};
  \draw (2.5,-1) node {$p$};
  \draw (2.5,0.75) node {$n+1$};
  \fill (0,0.1) rectangle (1.25,0.6);
  \node at (.625,0) [dot] {};
  \draw (.625,-1) node {$\frac{p}{2} - \frac{d}{4}$};
  \draw (1.25,-0.5) node {$p - \frac{d}{2}$};
  \draw (.625,0.75) node {$n+2$};
  \fill (3.75,0.1) rectangle (10,0.6);
  \node at (6.875,0) [dot] {};
  \draw (6.875,-1) node {$\half + \frac{p}{2} + \frac{d}{4}$};
  \draw (3.75,-0.5) node {$p + \frac{d}{2}$};
  \draw (6.875,0.75) node {$n+3$};
\end{tikzpicture}
\caption{Depiction of the deployment and radii assignment of sensors 
$n+1$, $n+2$, and $n+3$.}
\label{fig:partition}
\end{figure}

The following results are implied by Lemmas~\ref{lemma:bcfr-hard-half}
and~\ref{lemma:bcfr-hard-p}.

\begin{corollary}
There is no polynomial time multiplicative approximation algorithm for
\bcfr, unless P$=$NP, for any $a > 0$ and $\alpha \geq 1$, even if $x =
p^n$, where $p \in (0,1)$.
\end{corollary}

\begin{corollary}
There is no polynomial time algorithm that computes a solution within
an additive factor $\eps$, for some $\eps>0$, unless P$=$NP, for any
$a > 0$ and $\alpha \geq 1$, even if $x = p^n$, where $p \in (0,1)$.
\end{corollary}




\subsection{Variable Radii}

For \bcvr we show strong NP-hardness using a reduction from
\threepart%
\footnote{A \threepart instance consists of a list $a_1,\ldots,a_n$ of 
$n=3m$ positive integers such that $\frac{Q}{4} < a_i < \frac{Q}{2}$,
for every $i$, and $\sum_i a_i = mQ$, and the goal is to decide
whether the list can be partitioned into $m$ triples all having the
same sum $Q$.  \threepart remains NP-hard even if $Q$ is bounded above
by a polynomial in $n$.  In other words, the problem remains NP-hard
even when representing the integers in the input instance in unary
representation~\cite{GarJoh79}.}
that is based on the notion of \emph{block}, which is a set of evenly
spaced sensors with relatively small batteries.  A block battery
cannot move much, but together the block batteries can cover a long
interval, assuming they stay in their initial locations.  Formally, a
block $B = (z,\ell,b,\rho)$ is a set of $\ell$ sensors located at $z +
(2i-1)\rho$, for $i \in \set{1,\ldots,\ell}$.  The radius of each
block sensor is $\rho$, and the battery power of each sensor is $b$.
Typically, $\rho$ would be small, while $\ell$ would be large.


\begin{observation}
\label{obs:block}
Let $B = (z,\ell,b,\rho)$ be a block.  
\begin{inparaenum}[(i)]
\item $B$ can cover the interval $[z,z+2\ell\rho]$ for $b/\rho^\alpha$ 
      time, and 
\item no block sensor can cover points outside 
      $[z - \frac{b}{a},z + 2\ell\rho + \frac{b}{a}]$.
\end{inparaenum}
\end{observation}
\begin{proof}
If a block battery remains in its initial position, it can stay alive
for $b/\rho^\alpha$ time.  Since the batteries are at distance $2\rho$
from their neighbors, the interval $[z,z+2\ell \rho]$ is covered.
A sensor can move at most $b/a$, hence the leftmost and rightmost
point that can be reached by a block sensor are $z+\rho-b/a$ and
$z+2\ell\rho-\rho+b/a$.  Hence, no point outside
$[z-b/a,z+2\ell\rho+b/a]$ can be covered by a block sensor.
\qed
\end{proof}

We are now ready to present the reduction.


\begin{theorem}
\label{thm:hardness-bcvr}
\bcvr is strongly NP-hard, for every $a > 0$ and $\alpha \geq 1$.
\end{theorem}
\begin{proof}
Given an \bcvr instance and $T$, we show that it is NP-hard to
determine whether the instance can stay alive for $T$ time.  


Given a \threepart instance, we construct the following \bcvr
instance.  Let $\delta = \inv{(2m-1)Q}$ and $T =
2aQ[2(2m-1)Q]^\alpha$.  There is a sensor for each input number: $x_i
= 0$, and $b_i = T (a_i \delta/2)^\alpha + a$, for every $i \in
\set{1,\ldots,n}$.  We also add $m-1$ blocks: $B_j = ((2j-1) Q \delta,
\ceil{Q\delta/2\rho}, T \rho^\alpha, \rho)$, for every $j$, where
$\rho = \frac{\delta}{4} \cdot \inv{[2(2m-1)Q]^\alpha}$.

The running time of the reduction is polynomial, since each block
contains $O(m^{\alpha} Q^{\alpha+1})$ sensors, and there are $m-1$
blocks.

We show that if $(a_1,\dots,a_n) \in \threepart$, then there exists a
solution with lifetime $T$.
Since this instance belongs to \threepart, there is partition of
$\set{1,\ldots,n}$ into $m$ index subsets $I_1,\ldots,I_m$, such that
$|I_j|=3$ and $\sum_{i \in I_j} a_i = Q$, for any $j$.  We set $r_i =
\delta a_i/2$ for every $i \leq n$, and we deploy the sensors in $I_j$
such that they cover $[2jQ \delta,(2j+1)Q\delta]$.  Observe that the
three sensors in $I_j$ can cover the interval, since $\sum_{i \in I_j}
2r_i = \sum_{i \in I_j} a_i \delta = Q \delta$.  Also, each such
sensor uses at most $a$ energy for deployment, and hence it has enough
energy to stay alive for $T$ time.
Block sensors are not moved and their radii are set to $\rho$.  Hence,
block sensors can stay alive for $T$ time.  Furthermore, due to
Observation~\ref{obs:block}, the sensors of block $j$ can cover the
interval $[(2j-1)Q\delta,(2j-1)Q\delta + 2
\rho\ceil{Q\delta/(2\rho)}]$ during their lifetime.  Observed that
this interval contains $[(2j-1)Q\delta,2jQ\delta]$.  Hence $[0,1]$ can
be covered for $T$ time.

Now supposed that there is a solution with lifetime $T$.  It follows
that the block sensors radii cannot be larger than $\rho$.  Hence,
Observation~\ref{obs:block} implies that the sensors of block $j$ do
not cover points outside 
\[
[(2j-1)Q\delta - T\rho^\alpha/a, 
 (2j-1)Q\delta + 2\rho\ceil{Q\delta/(2\rho)} + T\rho^\alpha/a]
~.
\]
Since 
\[\textstyle
T \rho^\alpha/a 
=    2Q[2(2m-1)Q]^\alpha \cdot \frac{\delta^\alpha}{4^\alpha} 
     \cdot [2(2m-1)Q]^{-2\alpha}
\leq \half Q \delta \cdot [2(2m-1)Q]^{-\alpha}
\leq \frac{\delta}{8}
~
\]
and
\[\textstyle
2\rho 
=    \frac{2\delta}{4} \cdot \inv{[2(2m-1)Q]^\alpha}
\leq \frac{\delta}{8}
~,
\]
we have that the sensors of block $j$ do not cover points outside
$[(2j-1)Q\delta - \frac{\delta}{8}, 2jQ\delta + \frac{\delta}{4}]$.
It follows that the interval $[2jQ\delta + \frac{\delta}{4},
(2j+1)Q\delta - \frac{\delta}{8}]$ must be covered by a subset of the
first $n$ sensors whose sum of radii is at least
$(Q-\frac{3}{8})\delta$.

Since
\[
T(a_i \delta/2)^\alpha
=    2aQ[2(2m-1)Q]^\alpha a_i^\alpha [2(2m-1)Q]^{-\alpha} 
=    2aQ a_i^\alpha 
~,
\]
we have that the battery power of sensor $i$ is
\[\textstyle
b_i 
=    2aQ a_i^\alpha + a 
\leq 2aQ a_i^\alpha \cdot \frac{2Q+1}{2Q} 
\leq T(a_i \delta/2)^\alpha \cdot (\frac{2Q+1}{2Q})^\alpha
~.
\]
Hence, the radius that can be maintained by sensor $i$ for $T$ time
is at most $\frac{a_i \delta}{2} \cdot \frac{2Q+1}{2Q}$.  Since $a_i <
Q/2$, this radius is smaller that $\delta Q$, and therefore the $n$
sensors can be partitioned into $m$ subsets $I_1,\ldots,I_m$, each
covering an interval of length $(Q-\frac{3}{8})\delta$.
We claim that $\sum_{i \in I_j} a_i \geq Q$ for every subset $j$.
If this is not the case, then $\sum_{i \in I_j} a_i \leq Q-1$, for
some $j$.  Hence,
\[\textstyle
\sum_{i \in I_j} \frac{a_i \delta}{2} \cdot \frac{2Q+1}{2Q}
\leq (Q-1) \delta \cdot \frac{2Q+1}{2Q}
=    \frac{2Q^2 - Q - 1}{2Q} \cdot \delta
<    (Q - \half) \delta
<    (Q -\frac{3}{8})\delta
~.
\]
Hence, we can partition $a_1,\ldots,a_n$ into $m$ subsets each of sum
at least $Q$, which means that $(a_1,\ldots,a_n) \in \threepart$.
\qed 
\end{proof}